\documentclass[letterpaper, 10 pt, conference]{ieeeconf} 
\IEEEoverridecommandlockouts  
\usepackage{amsmath}
\usepackage{amssymb}
\usepackage{amsthm}
\usepackage{mathtools}
\usepackage{derivative}
\usepackage{xcolor}
\usepackage{bm}

\usepackage[noadjust]{cite}
\usepackage{url}

\usepackage{graphicx}
\usepackage[export]{adjustbox} % For aligning figures
\graphicspath{{figures/}}

\newtheorem{theorem}{Theorem}
\newtheorem{lemma}{Lemma}
\newtheorem{proposition}{Proposition}

\theoremstyle{definition}
\newtheorem{definition}{Definition}
\newtheorem{assumption}{Assumption}
\newtheorem{remark}{Remark}

\newtheorem{example}{Example}

\newcommand{\argmin}{\operatornamewithlimits{arg\,min}}

\newcommand{\R}{\mathbb{R}}

\newcommand{\Sc}{\mathcal{S}}

\newcommand{\T}{^\top}

\newcommand{\K}{\mathcal{K}}

\newcommand{\setdefb}[2]{\big\{#1 \; | \; #2\big\}}

\newcommand{\map}[3]{#1:#2 \rightarrow #3}
\newcommand{\bzero}{\mathbf{0}}

\newcommand{\bd}{\mathbf{d}}

\renewcommand{\bf}{\mathbf{f}} % NOTE: use \textbf if you really want to write bold non-math text.
\newcommand{\bg}{\mathbf{g}}

\newcommand{\bk}{\mathbf{k}}

\newcommand{\bu}{\mathbf{u}}
\newcommand{\bv}{\mathbf{v}}
\newcommand{\bw}{\mathbf{w}}
\newcommand{\bx}{\mathbf{x}}

\newcommand{\bz}{\mathbf{z}}

\newcommand{\bI}{\mathbf{I}}

\newcommand{\Dc}{\mathcal{D}}

\newcommand{\boldeta}{\bm{\eta}}
\newcommand{\bfeta}{\bm{\eta}}

\newcommand{\bpsi}{\bm{\psi}}

\newcommand{\bphi}{\bm{\phi}}

\usepackage{xcolor}
\definecolor{myblue}{RGB}{49, 114, 174}
\definecolor{myred}{rgb}{0.796, 0.235, 0.2}
\definecolor{mygreen}{rgb}{0.22, 0.596, 0.149}
\definecolor{mypurple}{rgb}{0.584,0.345,0.698}

\usepackage{blindtext}

\usepackage{tikz}

\newcommand\copyrighttext{%
  \footnotesize \textcopyright 2026 IEEE. Personal use of this material is permitted. Permission from IEEE must be obtained for all other uses, in any current or future media, including reprinting/republishing this material for advertising or promotional purposes, creating new collective works, for resale or redistribution to servers or lists, or reuse of any copyrighted component of this work in other works.}

\newcommand\copyrightnotice{%
\begin{tikzpicture}[remember picture,overlay]
\node[anchor=south,yshift=10pt] at (current page.south) {\parbox{\dimexpr\textwidth-\fboxsep-\fboxrule\relax}{\centering \copyrighttext}};
\end{tikzpicture}%
}

\title{\textbf{Input-to-State Safe Backstepping: \\ Robust Safety-Critical Control with Unmatched Uncertainties}}

\author{Max H. Cohen$^1$, Pio Ong$^2$, and Aaron D. Ames$^2$ %
\thanks{$^1$The author is with the Department of Electrical and Computer Engineering, North Carolina State University, Raleigh, NC \texttt{\{mhcohen2\}@ncsu.edu}.}
\thanks{$^2$The authors are with the Department of Mechanical and Civil Engineering, California Institute of Technology, Pasadena, CA \texttt{\{pioong,ames\}@caltech.edu}.}
}

\begin{document}
\maketitle
\copyrightnotice
\begin{abstract}
    Guaranteeing safety in the presence of unmatched disturbances---uncertainties that cannot be directly canceled by the control input---remains a key challenge in nonlinear control. This paper presents a constructive approach to safety-critical control of nonlinear systems with unmatched disturbances. We first present a generalization of the input-to-state safety (ISSf) framework for systems with these uncertainties using the recently developed notion of an Optimal Decay CBF,
    which
    provides more flexibility for satisfying the associated Lyapunov-like conditions for safety.
    From there, we outline a procedure for constructing ISSf-CBFs for two relevant classes of systems with unmatched uncertainties: i) strict-feedback systems; ii) dual-relative-degree systems, which are similar to differentially flat systems. Our theoretical results are illustrated via numerical simulations of an inverted pendulum and planar quadrotor.
\end{abstract}

\section{Introduction}

Control barrier functions (CBFs) \cite{AmesTAC17} have proven to be a practical tool for controlling safety-critical autonomous systems.
Given that autonomous systems deployed in the real world will inevitability operate under various forms of uncertainty, there has been large body of work on safety-critical control of uncertain systems. In particular, dynamic uncertainties arising from imperfect models or exogenous disturbances have been considered within safety-critical control frameworks by introducing various classes CBFs such as robust CBFs \cite{jankovic2018robust}, adaptive CBFs \cite{LopezLCSS21}, and input-to-state safe (ISSf) CBFs \cite{AmesLCSS19,AnilLCSS22}. Under various hypotheses regarding the class of uncertainties considered, these CBFs facilitate designing controllers guaranteeing safety in the presence of such uncertainties. While it is often straightforward to \emph{define} new types of CBFs that ensure safety in the presence of a particular class of uncertainty, it is typically more challenging to \emph{construct} a CBF meeting these definitions. 

There exist various techniques, such as high-order CBFs \cite{WeiTAC22,TanTAC22}, for constructing CBFs for systems without uncertainties; however, the ability to directly leverage such CBFs for uncertain systems depends heavily on the structure of these uncertainties. When the uncertainties are \emph{matched} -- uncertainties that, if known, could be directly canceled by the control input -- any CBF designed for a system without uncertainties may generally serve as a robust/adaptive/ISSf CBF for a system with uncertainties. On the other hand, when uncertainties are \emph{unmatched}, the synthesis of CBFs becomes more challenging. Despite these challenges, recent works have successfully developed various CBF-based control strategies that compensate for different classes of unmatched uncertainties \cite{LopezACC23,ErsinTAC25,KrsticACC25,WangTAC25,XuACC23,WangArXiV25}. 

In the context of stabilization, backstepping \cite{Krstic,KrsticDeng} has traditionally been used to design robust and adaptive controllers that successfully compensate for unmatched uncertainties. Recently, backstepping has been extended to safety-critical control with CBFs \cite{AndrewCDC22,CohenARC24}, enabling the systematic construction of CBFs for various classes of systems. Despite the precedent of leveraging backstepping to handle unmatched uncertainties \cite{Krstic,KrsticDeng}, results regarding CBF backstepping with unmatched uncertainties are scarcely found in the existing literature, the exception being \cite{KrsticACC25}. It should be noted, however, that \cite{KrsticACC25} leverages a different backstepping procedure than that from \cite{AndrewCDC22}, on which our present method is based (cf. \cite{KrsticLCSS25} for a comparison).

In this paper, we present a general framework for robust safety-critical control of nonlinear systems with unmatched uncertainties by uniting tools from ISSf \cite{AmesLCSS19} and CBF backstepping \cite{AndrewCDC22}. To this end, we first recount the notion of ISSf and discuss how matching conditions facilitate the verification of ISSf-CBF conditions. Next, we illustrate how the recently developed notion of an \emph{Optimal Decay} CBF (OD-CBF), a term originally coined in \cite{ZengACC21} and studied more formally in \cite{PioCDC25}, relaxes traditional CBF verification conditions, which becomes particularly useful when verifying robust versions of CBFs. When such conditions cannot be easily verified, we present a constructive procedure based on CBF backstepping \cite{AndrewCDC22,CohenARC24} that systematically synthesizes CBFs for certain classes of systems, even in the presence of unmatched uncertainties. The two classes of systems we consider are: 1) strict feedback systems, which are traditionally studied in backstepping \cite{Krstic} and 2) dual-relative-degree systems, as recently studied in \cite{BahatiCDC25}, which have similar properties to differentially flat systems \cite{MellingerICRA2011}, such as unicycles and quadrotors. Overall, this procedure allows for robust safety-critical control for practically relevant classes of systems in the presence of unmatched uncertainties, which we illustrate through the application of our developed theory via numerical simulations of a quadrotor and a pendulum.

In summary, the main contributions of this paper are: i) A generalization of ISSf framework via OD-CBFs, together with a broader disturbance-input model that provides additional flexibility in compensating for disturbances; ii) A constructive backstepping procedure for synthesizing ISSf-CBFs under unmatched disturbances, with a full technical treatment for strict-feedback and ``dual relative degree" systems.

\section{Background}
% \subsection{Control Barrier Functions}
\noindent\textbf{Control Barrier Functions:}
Consider the nonlinear control affine system:
\begin{equation}\label{sys:ctrl_affine}
    \dot \bx = \bf(\bx)+\bg(\bx)\bu,
\end{equation}
where $\bx\in\R^n$ is the system state and $\bu\in\R^m$ is the control input with $\map{\bf}{\R^n}{\R^n}$ and $\map{\bg}{\R^n}{\R^{n\times m}}$ smooth functions. Given a smooth feedback controller $\map{\bk}{\R^n}{\R^m}$, with $\bu=\bk(\bx)$, we obtain the closed-loop system:
\begin{equation}\label{sys:closed_loop}
    \dot \bx = \bf_{\mathrm{cl}}(\bx) \coloneq \bf(\bx)+\bg(\bx)\bk(\bx).
\end{equation}
For each initial condition, \eqref{sys:closed_loop} admits a solution defined on a maximal interval of existence, assumed to be the positive real line for ease of exposition. We say that a set $\Sc\subset \R^n$ is \emph{safe} if $\Sc$ is forward invariant for the closed-loop system. In this paper, we focus on candidate safe sets of the form:
\begin{subequations}
\label{eq:safeset}
\begin{align}
    \Sc = \setdefb{\bx\in\R^n}{h(\bx)\geq 0},\\
    \partial\Sc = \setdefb{\bx\in\R^n}{h(\bx)= 0},
\end{align}
\end{subequations}
where $\map{h}{\R^n}{\R}$ is continuously differentiable. Given a candidate safe set as in \eqref{eq:safeset}, we define $ b \coloneqq -\inf_{\bx\in\R^n}h(\bx)$, $c \coloneqq \sup_{\bx\in\R^n}h(\bx)$,
and the open set:
\begin{equation}\label{eq:open-D}
    \mathcal{D} \coloneqq \{\bx\in\R^n\mid h(\bx) + b > 0\}.
\end{equation}
The notion of a control barrier function allows one to design controllers enforcing safety.

 \begin{definition}[\cite{AmesTAC17}]
     A continuously differentiable function $\map{h}{\R^n}{\R}$ defining a set $\Sc\subset \R^n$ as in \eqref{eq:safeset} is said to be a \textit{control barrier function} (CBF) for~\eqref{sys:ctrl_affine} on the set $\mathcal{D}\supset\mathcal{S}$ from \eqref{eq:open-D} if there exists an extended class $\K$ function\footnote{A function $\map{\alpha}{(-b,c)}{\R}$ with $b,c>0$ is of extended class-$\K$ if it is continuous, strictly increasing, and $\alpha(0)=0$. 
     } $\alpha\,:\,(-b,c)\rightarrow\R$ such that\footnote{The Lie derivative of a continuously differentiable function $\map{h}{\R^n}{\R}$ along a vector field $\map{\bf}{\R^n}{\R^n}$ is denoted as $L_\bf h(\bx) = \nabla h(\bx)\bf(\bx)$.} for all $\bx\in\mathcal{D}$:
     \begin{equation}\label{eq:CBF_condition}
         \sup_{\bu\in\R^m} \big\{L_\bf h(\bx) +L_\bg h(\bx) \bu\big\} > -\alpha(h(\bx)).
     \end{equation}
 \end{definition}
Condition \eqref{eq:CBF_condition} ensures there is always a corrective input to increase $h$ at the boundary ($h(\bx)=0$) and prevents trajectories from leaving the safe set. Moreover, a CBF ensures the well-posedness and the regularity of the following optimization-based controller (CBF-QP):
\begin{align}
\label{eq:CBF-QP}
    \bk(\bx)  =\argmin_{\bu\in\R^m} \quad & \|\bu-\bk_{\rm{d}}(\bx)\|^2                                                                  \\
             \textup{s.t.} \quad &  L_\bf h(\bx) +L_\bg h(\bx) \bu \geq -\alpha(h(\bx)), \nonumber
\end{align}
which renders $\Sc$ safe for~\eqref{sys:closed_loop}.\footnote{ 
See \cite[Remark 1]{CohenARC24} for comments on the strict vs nonstrict inequality in \eqref{eq:CBF_condition} and \eqref{eq:CBF-QP}, respectively.}
The effectiveness of this controller is contingent on $h$ being a CBF -- the following lemma provides a convenient method to check if $h$ is a CBF:
\begin{lemma}[\cite{jankovic2018robust}]
    A continuously differentiable function $\map{h}{\R^n}{\R}$ is a CBF for \eqref{sys:ctrl_affine} on $\mathcal{D}\supset\mathcal{S}$ iff there exists an extended class $\mathcal{K}$ function $\alpha\,:\,(-b,c)\rightarrow\R$ such that:
    \begin{equation}\label{eq:CBF_verify}
        L_\bg h(\bx)=\bzero \implies L_\bf h(\bx)>-\alpha(h(\bx)),\quad \forall \bx\in\mathcal{D}.
    \end{equation}
\end{lemma}
The preceding lemma illustrates that verifying a CBF requires only checking the satisfaction of~\eqref{eq:CBF_condition} at points where $L_\bg h(\bx)=\bzero$.
This requirement remains unchanged for robust versions of CBFs, such as \textit{input-to-state-safe CBFs} \cite{AmesLCSS19,AnilLCSS22}, when \textit{matched disturbances} are introduced into the system\footnote{Note that \eqref{eq:CBF_verify} is \emph{equivalent} to \eqref{eq:CBF_condition} when the set of admissible control inputs is $\R^m$. If input bounds are present, i.e., $\bu\in\mathcal{U}\subset\R^m$, then \eqref{eq:CBF_verify} is still a necessary, but not a sufficient, condition for $h$ to be a CBF.}. 

% \subsection{Input-to-State-Safe CBFs}
\noindent\textbf{Input-to-State-Safe CBFs:}
We now consider systems subject to disturbances $\bd\in\R^n$: 
\begin{equation}\label{sys:ctrl_affine_general_disturbed}
    \dot \bx = \bf(\bx)+\bg(\bx)\bu + \bd.
\end{equation}
Given a smooth feedback controller $\bu=\bk(\bx)$ and a piecewise continuous disturbance signal $t\mapsto\bd(t)$, we obtain the disturbed closed-loop system:
\begin{equation}\label{sys:closed_loop_general_disturbed}
    \dot \bx = \bf(\bx)+\bg(\bx)\bk(\bx)+\bd(t),
\end{equation}
which also admits a unique solution defined on a maximal interval of existence.
In the presence of disturbances, one would expect a potential loss of safety properties. The notion of \emph{input-to-state safety} (ISSf) characterizes how safety degrades in the presence of disturbances.
\begin{definition}[\cite{AmesLCSS19}]
    A set $\Sc$ as in~\eqref{eq:safeset} is said to be \textit{input to state safe} (ISSf) for the closed-loop system~\eqref{sys:closed_loop_general_disturbed} if there exists a class $\K$ function $\gamma\,:\,[0, a)\rightarrow\R$, $a>0$, satisfying $\lim_{r\rightarrow a}\gamma(r)=b$ and a constant $\delta\in[0, a)$ such that, for any disturbance signal $t\mapsto\bd(t)$ satisfying $\|\bd(t)\|_\infty \leq \delta$, the set:    \begin{equation}\label{eq:safeset_enlarged}
        \Sc_\delta \triangleq\setdefb{\bx\in\R^n}{h(\bx) +\gamma(\delta) \geq 0},
    \end{equation}
    is forward invariant and hence safe.
\end{definition}
Within the ISSf framework, the set certified as safe is not the original candidate safe set $\mathcal{S}$, but an inflation of $\mathcal{S}$, denoted by $\mathcal{S}_{\delta}$ that grows proportionally to the magnitude of the disturbance, quantifying safety degradation. One benefit of the ISSf framework is the ability to systematically construct controllers enforcing ISSf using ISSf-CBFs.
\begin{definition}[\cite{AndrewTaylor}]
    A continuously differentiable function $\map{h}{\R^n}{\R}$ defining a set $\Sc\subset \R^n$ as in \eqref{eq:safeset} is said to be an \textit{ISSf control barrier function} (ISSf-CBF) for~\eqref{sys:ctrl_affine_general_disturbed} if there exists an extended class $\K$ function $\alpha\,:\,(-b,c)\rightarrow\R$ and a positive constant $\varepsilon> 0$ such that for all $\bx\in\R^n$: \begin{equation}\label{eq:ISSf-CBF_condition}
         \sup_{\bu\in\R^m} \big\{L_\bf h(x) +L_\bg h(\bx) \bu\big\} > -\alpha(h(\bx))+\tfrac{1}{\varepsilon}\|\nabla h(\bx)\|^2.
     \end{equation}
\end{definition}

Analogous to standard CBFs, the CBF-QP like~\eqref{eq:CBF-QP}, modified to enforce the ISSf-CBF condition~\eqref{eq:ISSf-CBF_condition} ensures $\Sc$ is ISSf \cite{AmesLCSS19,AndrewTaylor}. 
Similar to CBFs, determining if a function is an ISSf-CBF requires studying its behavior when $L_\bg h(\bx)=\bzero$.
\begin{lemma}[\cite{AndrewTaylor}]
    A continuously differentiable function $\map{h}{\R^n}{\R}$ is an ISSf-CBF for \eqref{sys:ctrl_affine_disturbed}  on $\Dc\supset \Sc$ if and only if there exists an extended class $\K$ function $\map{\alpha}{(-b,c)}{\R}$ such that for each $\bx\in\Dc$:
    \begin{equation}\label{eq:ISSf-CBF_verify}
        L_\bg h(\bx)=\bzero \implies L_\bf h(\bx) > -\alpha(h(\bx)) + \tfrac{1}{\varepsilon}\|\nabla h(\bx)\|^2.
    \end{equation}
\end{lemma}
Compared to standard CBFs, the verification condition for ISSf-CBFs includes the additional robust term $\frac{1}{\varepsilon}\|\nabla h(\bx)\|^2$ in~\eqref{eq:ISSf-CBF_verify}, which makes the verification process more stringent. This motivates the development of methods that relax the verification condition, as we discuss in the next section.

\section{Optimal Decay Input-to-State Safe CBFs}
A key challenge in leveraging the ISSf framework for robust safety-critical control lies in verifying whether the associated CBF condition holds in the presence of disturbances. As shown in~\eqref{eq:ISSf-CBF_verify}, the natural dynamics of the system must satisfy a stricter requirement at states where the control authority vanishes, which may not hold in general. 

In this section, we address this via two generalizations. First, we consider a generalized disturbance model:
\begin{equation}\label{sys:ctrl_affine_disturbed}
    \dot \bx = \bf(\bx)+\bg(\bx)\bu+\bw(\bx)\bd,
\end{equation}
where $\map{\bw}{\R^n}{\R^{n\times p}}$ is the disturbance input matrix, which characterizes the directions in which the disturbance enters the dynamics. If such directions are unknown, one could take $\bw(\bx)=\bI$ to recover~\eqref{sys:ctrl_affine_general_disturbed}. Second, we adopt the notion of an \emph{optimal decay} CBF (OD-CBF), introduced in \cite{ZengACC21} and more formally characterized in \cite{PioCDC25}.
Now, we extend this notion to the ISSf setting as follows.

\begin{definition}
    A continuously differentiable function $h\,:\,\R^n\rightarrow\R$ defining a set $\mathcal{S}\subset\R^n$ as in \eqref{eq:safeset} is said to be an \emph{optimal-decay} ISSf-CBF (OD-ISSf-CBF) for \eqref{sys:ctrl_affine_disturbed} on the set $\mathcal{D}\supset\mathcal{S}$ from \eqref{eq:open-D} if there exists an extended class $\K$ function $\alpha\,:\,(-b,c)\rightarrow\R$ and positive constants $\theta_{\rm{d}}>0$ and $\varepsilon>0$ such that for all $\bx\in\mathcal{D}$:
    \begin{equation}\label{eq:OD-ISSf-CBF}
        \sup_{\substack{\bu\in\R^m \\ \omega\geq\theta_{\rm{d}}}} \{L_{\bf}h(\bx) + L_{\bg}h(\bx)\bu + \omega\alpha(h(\bx))\} > \tfrac{1}{\varepsilon}\|L_{\bw}h(\bx)\|^2.
    \end{equation}
\end{definition}
Compared to conventional CBFs, OD-CBFs introduce an additional decision variable, denoted by $\omega$ in \eqref{eq:OD-ISSf-CBF}, that acts as a scaling factor on the extended class $\mathcal{K}$ function $\alpha$. This scaling factor can be viewed as an additional control input that provides more flexibility in satisfying the associated barrier condition. Compared to the definition in \cite{PioCDC25}, wherein an OD-CBF was defined only on $\mathcal{S}$ with $\omega\geq0$, here we require $\omega\geq\theta_{\rm{d}}>0$ to extend these ideas to the set $\mathcal{D}$ containing $\mathcal{S}$. A benefit of OD-CBFs is that their validity need only be verified on a smaller region compared to CBFs.
\begin{lemma}\label{lemma:OD-ISSf-verify}
    A continuously differentiable function $h\,:\,\R^n\rightarrow\R$ is an OD-ISSf-CBF for \eqref{sys:ctrl_affine_disturbed} on the set $\mathcal{D}\supset\mathcal{S}$ if and only if there exists an extended class $\K$ function $\map{\alpha}{(-b,c)}{\R}$ such that for each $\bx\in\Dc$:
    \begin{multline}\label{eq:OD-ISSf-CBF-verify}
            h(\bx) \leq 0\;\wedge\;L_{\bg}h(\bx) = \bzero  \\ 
             \implies
            L_{\bf}h(\bx) + \theta_{\rm{d}}\alpha(h(\bx)) > \frac{1}{\varepsilon}\|L_{\bw}h(\bx)\|^2.
    \end{multline}
\end{lemma}
\begin{proof}
    The proof is similar to that of \cite[Lemma 3]{PioCDC25}. Define:
    \begin{equation}
        \begin{aligned}
            \rho(\bx,\bu,\omega) = & L_{\bf}h(\bx) + L_{\bg}h(\bx)\bu \\ & + \omega\alpha(h(\bx)) - \frac{1}{\varepsilon}\|L_{\bw}h(\bx)\|^2.
        \end{aligned}
    \end{equation}
    Provided that $h(\bx)>0$ or $L_{\bg}h(\bx)\neq\bzero$, we have:
    \begin{equation*}
        \sup_{\bu\in\R^m,~  \omega\geq\theta_{\rm{d}}} \rho(\bx,\bu,\omega) = \infty,
    \end{equation*}
    implying that \eqref{eq:OD-ISSf-CBF} holds. It remains to check the points in $\mathcal{D}$ where $h(\bx)\leq 0$ and $L_{\bg}h(\bx)=\bzero$. At such points:
    \begin{equation*}
        \rho(\bx,\bu,\omega) = L_{\bf}h(\bx) + \omega\alpha(h(\bx)) - \tfrac{1}{\varepsilon}\|L_{\bw}h(\bx)\|^2.
    \end{equation*}
    Since $\alpha$ is an extended class $\mathcal{K}$ function and $h(\bx)\leq0$, the composition $\alpha(h(\bx))\leq0$, so taking the supremum yields:
    \begin{equation*}
        \sup_{\substack{\bu\in\R^m \\ \omega\geq\theta_{\rm{d}}}} \rho(\bx,\bu,\omega) = L_{\bf}h(\bx) + \theta_{\rm{d}}\alpha(h(\bx)) - \tfrac{1}{\varepsilon}\|L_{\bw}h(\bx)\|^2.
    \end{equation*}
    Hence, asking if \eqref{eq:OD-ISSf-CBF} holds for all $\bx\in\mathcal{D}$ is equivalent to asking if \eqref{eq:OD-ISSf-CBF-verify} holds, as desired. 
\end{proof}

Both generalizations introduced herein relax the verification process from~\eqref{eq:ISSf-CBF_verify} to~\eqref{eq:OD-ISSf-CBF-verify}. First, as a result of the optimal-decay variable, one need only check the validity of 
$h$
on the boundary and outside $\mathcal{S}$. On the interior of $\mathcal{S}$, \eqref{eq:OD-ISSf-CBF} is automatically satisfied. A consequence of this is that $h$ is an OD-ISSf-CBF provided that $L_{\bg}h(\bx)\neq\bzero$ on $\mathcal{D}\setminus\mathrm{Int}(\mathcal{S})$. 

\begin{proposition}\label{prop:Lgh=0}
    A continuously differentiable function $h\,:\,\R^n\rightarrow\R$ is an OD-ISSf-CBF for \eqref{sys:ctrl_affine_disturbed} on the set $\mathcal{D}\supset\mathcal{S}$ if:
    \begin{equation}
        L_{\bg}h(\bx)\neq \bzero,\quad \forall \bx\in\mathcal{D}\setminus\mathrm{Int}(\mathcal{S}).
    \end{equation}
\end{proposition}

The second relaxation in~\eqref{eq:OD-ISSf-CBF-verify} comes from the disturbance input matrix $\bw(\bx)$. The robustness term now depends on $L_\bw h(\bx)$, which is a projection of $\nabla h(\bx)$ onto the directions specified by $\bw(\bx)$, rather than the full gradient. This reduces conservatism by requiring robustness only along the directions where disturbances affect the function $h$.

This observation connects to the notion of \emph{matched disturbances}, which played a central role in the development of ISSf-CBFs. In particular, early formulations~\cite{AmesLCSS19} used $\|L_\bg h(\bx)\|^2$ as the robustness term, corresponding to disturbances that enter the system through the same channels as the control inputs. We use the following formal definition.
\begin{definition}[\cite{LopezLCSS21Contraction}]
    A disturbance input $\bd\in \R^p$ for~\eqref{sys:ctrl_affine_disturbed} is said to be \textit{matched} if $\bw(\bx)\bd\in\mathrm{span}(\bg(\bx))$. Otherwise, $\bd$
    is said to be \emph{unmatched}.
\end{definition}
The matching condition implies $\bw$ can be written as a linear combination of the columns of $\bg$ in that there exists a function $\bphi\,:\,\R^n\rightarrow\R^{m\times p}$ such that $\bw(\bx)=\bg(\bx)\bphi(\bx)$, cf. \cite{LopezLCSS21Contraction}.  
The main benefit of the matching condition is that: 
\begin{equation*}
    L_{\bg}h(\bx)=\bzero \implies L_{\bw}h(\bx) = L_{\bg}h(\bx)\bphi(\bx) = \bzero,
\end{equation*}
so that \eqref{eq:OD-ISSf-CBF-verify} reduces to the standard OD-CBF verification. Thus, if 
$\bd$
is matched, a CBF can be constructed without considering disturbances, and then automatically used within an ISSf controller to robustify the system to disturbances.

When $h$ is an OD-ISSf-CBF, we may construct a feedback controller using the CBF-based QP framework:
\begin{equation}\label{eq:OD-ISSf-CBF-QP}
    \begin{aligned}
        \begin{bmatrix}
            \bk(\bx)\\\theta(\bx)
        \end{bmatrix} &= \argmin_{\substack{\bu\in\R^m \\ \omega\in\R}}  \quad \tfrac{1}{2}\|\bu - \bk_{\rm{d}}(\bx)\|^2 + \tfrac{1}{2}p(\omega - \theta_{\rm{d}})^2 \\ 
        \mathrm{s.t.} & \quad L_{\bf}h(\bx) + L_{\bg}h(\bx)\bu \geq - \omega\alpha(h(\bx)) + \tfrac{1}{\varepsilon}\|L_{\bw}h(\bx)\|^2 \\
        & \quad \omega \geq \theta_{\rm{d}}
    \end{aligned}
\end{equation}
where $p>0$. As in \cite{ZengACC21,PioCDC25}, the OD paradigm is centered around jointly optimizing the control input $\bu$ and a scaling factor $\omega$. Ultimately, the controller $\bk$ above ensures that:
\begin{equation}\label{eq:OD-ISSf-controller-enforces}
    L_{\bf}h(\bx) + L_{\bg}h(\bx)\bk(\bx) \geq - \theta(\bx)\alpha(h(\bx)) + \tfrac{1}{\varepsilon}\|L_{\bw}h(\bx)\|^2
\end{equation}
with a state-dependent scaling factor $\theta\,:\,\mathcal{D}\rightarrow\R_{>0}$  for $\alpha$, which satisfies $
    \theta(\bx) \geq \theta_{\rm{d}},\; \forall \bx\in\mathcal{D}
$
for a desired value $\theta_{\rm{d}}>0$. Following a similar approach to that in \cite[Lemma 4]{PioCDC25}, one may show that \eqref{eq:OD-ISSf-CBF-QP} is locally Lipschitz continuous on $\mathcal{D}$ and may be expressed in closed-form as:
\begin{equation}
    \begin{aligned}
        \bk(\bx) = & \bk_{\rm{d}}(\bx) + \lambda(\upsilon(\bx),\xi(\bx),\zeta(\bx))L_{\bg}h(\bx)\T \\ 
        \upsilon(\bx) = & L_{\bf}h(\bx) + L_{\bg}h(\bx)\bk_{\rm{d}}(\bx) + \theta_{\rm{d}}\alpha(h(\bx)) - \tfrac{\|L_{\bw}h(\bx)\|^2}{\varepsilon} \\ 
        \xi(\bx) = & \|L_{\bg}h(\bx)\| , \qquad\qquad
        \zeta(\bx) = (1/p) \alpha(h(\bx)), \\ 
    \end{aligned}
\end{equation}
where, with $\mathrm{ReLU}(\cdot) = \max(0,\cdot)$:
\begin{equation}
    \lambda(\upsilon,\xi,\zeta) = \begin{cases}
        0 & \text{if } \xi = 0\, \wedge\, \zeta\leq 0, \\ 
        \frac{\mathrm{ReLU}(-\upsilon)}{\xi^2 + p\mathrm{ReLU}(\zeta)^2} & \text{otherwise}.
    \end{cases}
\end{equation}
Moreover, the state dependent scaling factor on $\alpha$ is:
\begin{equation}\label{eq:theta-closed-form}
    \begin{aligned}
        \theta(\bx) = & \theta_{\rm{d}} + \psi(\upsilon(\bx),\xi(\bx),\zeta(\bx)), \\ 
        \psi(\upsilon,\xi,\zeta) = & \begin{cases}
        0 & \text{if } \xi = 0\, \wedge\, \zeta\leq 0, \\ 
        \frac{\mathrm{ReLU}(-\upsilon)\mathrm{ReLU}(\zeta)}{\xi^2 + pc^2} & \text{otherwise}.
        \end{cases}
    \end{aligned}
\end{equation}

Before stating the safety properties of this controller, we require the following lemma, which will facilitate the use of Nagumo's Theorem \cite[Ch. 4]{AbrahamMarsdenRatiu} in proving set invariance.
\begin{lemma}\label{lemma:regular-value}
    If $h$ is an OD-ISSf-CBF for \eqref{sys:ctrl_affine_disturbed} on the set $\mathcal{D}\supset\mathcal{S}$ as in \eqref{eq:open-D}, then any $\kappa\in(-b,0]$ is a regular value of $h$. 
\end{lemma}

\begin{proof}
    For $\kappa$ to be a regular value of $h$, we require $\nabla h(\bx)\neq \bzero$ whenever $h(\bx)=\kappa$. For the sake of contradiction, suppose there exists $\kappa\in(-b,0]$ such that $\kappa$ is not a regular value of $h$. Thus, when $h(\bx)=\kappa\leq0$ we have $\nabla h(\bx)=\bzero$ (implying $L_{\bf}h(\bx)=0,L_{\bg}h(\bx)=\bzero$), which, using the assumption that $h$ is an OD-ISSf-CBF \eqref{eq:OD-ISSf-CBF}, implies:
    \begin{equation*}
        \sup_{\substack{\bu\in\R^m \\ \omega\geq\theta_{\rm{d}}}} \omega\alpha(h(\bx)) = \sup_{\omega\geq\theta_{\rm{d}}} \omega\alpha(h(\bx)) = \sup_{\omega\geq\theta_{\rm{d}}} \omega\alpha(\kappa) > 0.
    \end{equation*}
    For any $\kappa\in(-b,0]$ we then reach the contradiction:
    \begin{equation*}
         \sup_{\omega\geq\theta_{\rm{d}}} \omega\underbrace{\alpha(\kappa)}_{\leq 0} = \underbrace{\theta_{d}}_{>0}\underbrace{\alpha(\kappa)}_{\leq 0}  > 0,
    \end{equation*}
    implying that $\kappa$ must be a regular value of $h$, 
    as desired.
\end{proof}

With the preceding lemma, we can now state the main result regarding safety properties of OD-ISSf-CBFs.

\begin{theorem}\label{theorem:OD-ISSf-safety}
    Let $h$ be an OD-ISSf-CBF for \eqref{sys:ctrl_affine_disturbed} on a set $\mathcal{D}\supset\mathcal{S}$ as in \eqref{eq:open-D} and consider the CBF-QP controller in \eqref{eq:OD-ISSf-CBF-QP}. Provided that:
    \begin{equation}\label{eq:ISSf-OD-safety-conditions}
        \|\bd\|_{\infty}^2 < -\frac{2\theta_{\rm{d}}\alpha(-b)}{\varepsilon},
    \end{equation}
    then the set $\mathcal{S}$ as in \eqref{eq:safeset} is ISSf  for the closed-loop system. In particular, the set $\mathcal{S}_{\delta}$ defined as in \eqref{eq:safeset_enlarged} is forward invariant for the closed-loop system, where:
    \begin{equation}\label{eq:gamma-ISSf-OD-safety}
        \gamma(\delta) = -\alpha^{-1}\left(-\frac{\varepsilon\delta^2}{2\theta_d} \right).
    \end{equation}
    Furthermore, if $\inf_{\bx\in\partial \Sc}\|L_\bw h(\bx)\| > \varepsilon \|\bd\|_\infty$, then the original set $\Sc$ is forward invariant. 
\end{theorem}
\begin{proof}
    Differentiating $h$ along solutions of the closed-loop disturbed system \eqref{sys:ctrl_affine_disturbed} yields:
    \begin{equation}\label{eq:hdot_ISSf_young}
        \begin{aligned}
            \dot{h} = & L_{\bf}h(\bx) + L_{\bg}h(\bx)\bk(\bx) + L_{\bw}h(\bx)\bd \\ 
            \geq & -\theta(\bx)\alpha(h(\bx)) + \frac{1}{\varepsilon}\|L_{\bw}h(\bx) \|^2 + L_{\bw}h(\bx)\bd \\ 
            \geq & -\theta(\bx)\alpha(h(\bx)) + \frac{1}{\varepsilon}\|L_{\bw}h(\bx) \|^2 - \|L_{\bw}h(\bx)\| \|\bd\| \\
            \geq & -\theta(\bx)\alpha(h(\bx)) + \frac{1}{\varepsilon}\|L_{\bw}h(\bx) \|^2 - \|L_{\bw}h(\bx)\| \|\bd\|_{\infty} \\ 
            \geq & -\theta(\bx)\alpha(h(\bx)) + \frac{1}{2\varepsilon}\|L_{\bw}h(\bx) \|^2- \frac{\varepsilon}{2}\|\bd\|^2_{\infty} 
        \end{aligned}
    \end{equation}
    where the first inequality follows from \eqref{eq:OD-ISSf-controller-enforces}, the second from taking norm-bounds, the third from $\|\bd\|\leq\|\bd\|_{\infty}$, and the fourth from applying Young's inequality on the last term. Based on \eqref{eq:theta-closed-form}, we have $\theta(\bx)=\theta_{\rm{d}}>0$ whenever $h(\bx)
    \leq0$. Hence, for all $\bx\in\mathcal{D}\setminus\mathrm{Int}(\mathcal{S})$, we have:
    \begin{equation}
        \dot{h} \geq - \theta_d\alpha(h(\bx)) - \frac{\varepsilon}{2}\|\bd\|_{\infty}^2.
    \end{equation}
    Let $h_{\delta}(\bx)=h(\bx) + \gamma(\|\bd\|_{\infty})$ for $\gamma\in\mathcal{K}$. Then:
    \begin{equation}
        \begin{aligned}
            \dot{h}_{\delta} = \dot{h} \geq- \theta_d\alpha(h_{\delta}(\bx) - \gamma(\|\bd\|_{\infty})) - \frac{\varepsilon}{2}\|\bd\|_{\infty}^2.
        \end{aligned}
    \end{equation}
    Provided that $\gamma$ is chosen as in \eqref{eq:gamma-ISSf-OD-safety}, we have:
    \begin{equation}\label{eq:h-delta-Nagumo}
        h_{\delta}(\bx) = 0 \implies \dot{h}_{\delta}(\bx,\bk(\bx),\bd)\geq0.
    \end{equation}
    Now, note that \eqref{eq:ISSf-OD-safety-conditions} ensures that $\gamma(\|\bd\|_{\infty})<b$. Based on \eqref{eq:open-D}, this ensures that $\mathcal{S}_{\delta}\subset\mathcal{D}$. 
    Furthermore, since $\kappa$ is a regular value of $h$ for any $\kappa\in(-b,0]$ by Lemma \ref{lemma:regular-value} and $\gamma(\|\bd\|_{\infty})<b$ it follows that $-\gamma(\|\bd\|_{\infty})$ is a regular value of $h$, which implies that zero is a regular value of $h_{\delta}$. Since \eqref{eq:h-delta-Nagumo} holds and zero is a regular value of $h_{\delta}$, it follows from Nagumo's Theorem \cite[Ch. 4]{AbrahamMarsdenRatiu} that $\mathcal{S}_{\delta}$ is forward invariant for the closed-loop system. Hence, $\mathcal{S}$ is ISSf, as claimed. In addition, the expression for $\dot h$ in~\eqref{eq:hdot_ISSf_young} is greater than zero on $\partial \Sc$  if $\inf_{\bx\in\partial \Sc}\|L_\bw h(\bx)\| > \varepsilon \|\bd\|_\infty$ holds, and forward invariance of $\Sc$ follows from OD-CBF result in~\cite{PioCDC25}.
\end{proof}

% To summarize, 
The OD-ISSf-CBF framework relaxes 
the verification process of
ISSf-CBFs
in two 
% important 
ways. First, the CBF condition only needs to be checked on the boundary and the exterior of $\Sc$. Second,  robustness is only required along the direction of $L_\bw h(\bx)$. In particular, if $L_{\bg}h(\bx)\neq\bzero$ for all $\bx\in\mathcal{D}\setminus\mathrm{Int}(\mathcal{S})$, then $h$ is automatically an OD-ISSf-CBF; otherwise, one must verify that \eqref{eq:OD-ISSf-CBF-verify} holds.
Choosing smaller $\varepsilon$ and larger $\theta_{\rm{d}}$ increases robustness in that larger disturbances can be compensated for, cf. \eqref{eq:ISSf-OD-safety-conditions}, and that the inflation of $\mathcal{S}$ is decreased, cf. \eqref{eq:gamma-ISSf-OD-safety}. One must take care, however, to ensure that these chosen parameters are compatible with the OD-ISSf-CBF in the sense that they satisfy \eqref{eq:OD-ISSf-CBF-verify}.

\section{Input-to-State Safe Backstepping}
In this section, we further refine our robust safety framework by no longer assuming that the system has direct control authority over the function $h$ or that the disturbance is matched. Instead, we show that, in many cases, the structure of the system can be exploited through backstepping~\cite{AndrewCDC22,CohenARC24} to construct CBFs that yield ISSf guarantees. 

\subsection{Strict-Feedback Systems}\label{sec:strict-feedback}
The first case we consider is systems where control authority over the safety function is not direct, but states relevant to this safety function can still be influenced through higher-order dynamics. A canonical representation of such systems is given by the strict-feedback form:
\begin{equation}\label{eq:strict-feedback}
    \underbrace{
    \begin{bmatrix}
        \dot{\bx}_1 \\ \dot{\bx}_2
    \end{bmatrix}}_{\dot{\bx}}
    =
    \underbrace{
    \begin{bmatrix}
        \bf_1(\bx_1) + \bg_1(\bx_1)\bx_2 \\ 
        \bf_2(\bx)
    \end{bmatrix}}_{\bf(\bx)}
    +
    \underbrace{
    \begin{bmatrix}
        \bzero \\ \bg_{2}(\bx)
    \end{bmatrix}}_{\bg(\bx)}
    \bu
    +
    \underbrace{
    \begin{bmatrix}
        \bw_1(\bx_1) \\ \bw_2(\bx)
    \end{bmatrix}}_{\bw(\bx)}
    \bd,
\end{equation}
where $\bx=(\bx_1,\bx_2)\in\R^{n_1}\times\R^{n_2}=\R^n$ is the system state. The objective is to design a feedback controller that ensures the system trajectory remains in the set:
\begin{equation}\label{eq:S1}
    \mathcal{S}_1 = \{\bx_1\in\R^{n_1}\mid h_1(\bx_1)\geq0\},
\end{equation}
for all time.  Backstepping proceeds by viewing the second state $\bx_2$ as a ``virtual" input to the dynamics of the first state:
\begin{equation}\label{eq:top-subsystem}
    \dot{\bx}_1 = \bf_1(\bx_1) + \bg_1(\bx_1)\bx_2 + \bw_1(\bx_1)\bd,
\end{equation}
and assuming that $h_1$ is an OD-ISSf-CBF for \eqref{eq:top-subsystem} on a set $\mathcal{D}_1\supset\mathcal{S}_1$ in the sense that:
\begin{align}\label{eq:OD-ISSf-CBF-virtual}
        \sup_{\bv\in\R^{n_2},~\omega\geq\theta_{\rm{d}}} \rho_1(\bx_1,&\bv,\omega) >   0, \\
        \rho_1(\bx_1,\bv,\omega) \coloneqq & L_{\bf_1}h_1(\bx_1) + L_{\bg_1}h_1(\bx_1)\bv \nonumber\\ &+ \omega\alpha(h_1(\bx_1)) -  \frac{1}{\varepsilon}\|L_{\bw_1}h_1(\bx_1)\|^2,\nonumber
\end{align}
for all $\bx_1\in\mathcal{D}_1$.
A key component of CBF backstepping is the existence of a smooth feedback controller $\bk_1\,:\,\mathcal{D}_1\rightarrow\R^{n_2}$ satisfying the associated barrier condition:
\begin{multline}\label{eq:OD-ISSf-controller-smooth}
        L_{\bf_1}h_1(\bx_1) + L_{\bg_1}h_1(\bx_1)\bk_1(\bx_1) \\ > - \theta_d\alpha(h_1(\bx_1)) + \tfrac{1}{\varepsilon}\|L_{\bw_1}h_1(\bx_1)\|^2,
\end{multline}
for all $\bx_1\in\mathcal{D}_1$. Note here that we have replaced the state-dependent $\theta$ with the value $\theta_{\rm{d}}$, which will facilitate the construction of smooth CBF-based controllers. General results regarding smoothness of CBF-based controllers are presented in \cite{CohenLCSS23}, and can be readily applied to design a smooth controller satisfying \eqref{eq:OD-ISSf-controller-smooth}, cf. \cite{CohenARC24,BahatiCDC25}.

With a smooth controller in place, we may extend the barrier condition from the top subsystem to the full system. This is done by penalizing the deviation of the virtual input $\bx_2$ from the desired safeguarding controller $\bk_1(\bx_1)$ as:
\begin{equation}\label{eq:h-backstepping}
    h(\bx) = h_1(\bx_1) - \tfrac{1}{2\mu}\|\bx_2 - \bk_1(\bx)\|^2,
\end{equation}
with a parameter $\mu>0$, as a candidate OD-ISSf-CBF on the set $\mathcal{D}\supset\mathcal{S}$ as defined in \eqref{eq:open-D}. Under a rank condition on $\bg_2$, this candidate is indeed an OD-ISSf-CBF.

\begin{theorem}
    Consider the strict-feedback system \eqref{eq:strict-feedback} and a constraint set $\mathcal{S}_1$ as in \eqref{eq:S1}. Suppose there exists a smooth controller $\bk_1\,:\,\mathcal{D}_1\rightarrow\R^{n_2}$ satisfying \eqref{eq:OD-ISSf-controller-smooth}, and construct $h$ as in~\eqref{eq:h-backstepping} with a set $\Dc\subseteq \Dc_1\times \R^{n_2}$ defined in~\eqref{eq:open-D}. If $\bg_2\,:\,\R^{n}\rightarrow\R^{n_2\times m}$ is full row rank for all $\bx\in\mathcal{D}\setminus\mathrm{Int}(\mathcal{S})$, then $h$ is an OD-ISSf-CBF for \eqref{eq:strict-feedback}.
\end{theorem}

\begin{proof}
    To show that \eqref{eq:h-backstepping} is an OD-ISSf-CBF, we must show that \eqref{eq:OD-ISSf-CBF-verify} holds. Computing gradients of $h$ yields:
    \begin{equation}
        \begin{aligned}
            \pdv{h}{\bx_1}(\bx) = & \pdv{h_1}{\bx_1}(\bx) - \frac{1}{\mu}(\bx_2 - \bk_1(\bx_1))\T \pdv{\bk_1}{\bx_1}(\bx_1) \\ 
            \pdv{h}{\bx_2}(\bx) = & -\frac{1}{\mu}(\bx_2 - \bk_1(\bx_1))\T,
        \end{aligned}
    \end{equation}
    which implies that:
    \begin{equation}
        \begin{aligned}
            L_{\bg}h(\bx) = &
        \begin{bmatrix}
            \pdv{h}{\bx_1}(\bx) & \pdv{h}{\bx_2}(\bx)
        \end{bmatrix}
        \begin{bmatrix}
            \bzero \\ \bg_2(\bx)
        \end{bmatrix} \\
        = &
        -\frac{1}{\mu}(\bx_2 - \bk_1(\bx_1))\T\bg_2(\bx).
        \end{aligned}
    \end{equation}
    If $\bg_2$ has full row rank on $\mathcal{D}\setminus\mathrm{Int}(\mathcal{S})$ then:
    \begin{equation}
        L_{\bg}h(\bx) = \bzero \iff \bx_2 = \bk_1(\bx_1),
    \end{equation}
    for all $\bx\in\mathcal{D}\setminus\mathrm{Int}(\mathcal{S})$. 
    Hence, at any point where $L_{\bg}h(\bx) = \bzero$ and $h(\bx)\leq0$, we have $h(\bx)= h_1(\bx_1)$ and:
    \begin{equation}\label{eq:equivalent-lie-derivatives}
        \begin{aligned} 
            \pdv{h}{\bx}(\bx) = & 
            \begin{bmatrix}
            \pdv{h_1}{\bx_1}(\bx) & \bzero
            \end{bmatrix} ~\qquad \implies\\ 
            L_{\bf}h(\bx) = & \pdv{h}{\bx}(\bx)\bf(\bx) = L_{\bf_1}h_1(\bx_1) + L_{\bg_1}h_1(\bx_1)\bk_1(\bx), \\ 
            L_{\bw}h(\bx) = & \pdv{h}{\bx}(\bx)\bw(\bx) = L_{\bw_1}h_1(\bx_1).
        \end{aligned}
    \end{equation}
    Thus, any point where $L_{\bg}h(\bx)=\bzero$ and $h(\bx)\leq0$, we have:
    \begin{equation}
        \begin{aligned}
            L_{\bf}h(\bx) = & L_{\bf_1}h_1(\bx_1) + L_{\bg_1}h_1(\bx_1)\bk_1(\bx_1) \\ 
            % > & -\theta_d\alpha(h_1(\bx_1)) + \frac{1}{\varepsilon}\|L_{\bw_1}h_1(\bx_1)\|^2  \\ 
            % = & \theta_{\rm{d}}\alpha(h_1(\bx_1)) + \frac{1}{\varepsilon}\|L_{\bw_1}h_1(\bx_1)\|^2 \\
            > & -\theta_{\rm{d}}\alpha(h(\bx)) + \frac{1}{\varepsilon}\|L_{\bw}h(\bx)\|^2,
        \end{aligned}
    \end{equation}
    where the inequality follows from \eqref{eq:OD-ISSf-controller-smooth} and \eqref{eq:equivalent-lie-derivatives}. From Lemma \ref{lemma:OD-ISSf-verify},
    this implies
    \eqref{eq:h-backstepping} 
    is an OD-ISSf-CBF for \eqref{eq:strict-feedback}.
\end{proof}

Although the results here have been applied to a system with two layers, the same procedure may be recursively applied to systems with an arbitrary number of layers:
\begin{equation*}
    \begin{aligned}
        \dot{\bx}_1 = &  \bf_1(\bx_1) + \bg_1(\bx_1)\bx_2 + \bw_1(\bx_1)\bd \\ 
        \dot{\bx}_2 = &  \bf_2(\bx_{1},\bx_2) + \bg_2(\bx_{1},\bx_2)\bx_3 + \bw_2(\bx_{1},\bx_2)\bd \\ 
         & \vdots \\
        \dot{\bx}_{N} = & \bf_{N}(\bx) + \bg_{N}(\bx)\bu + \bw_{N}(\bx)\bd.
    \end{aligned}
\end{equation*}

We emphasize that in the strict-feedback setting, disturbances are generally unmatched with respect to the control input~$\bu$. Rather than requiring $\bw$ to lie in the span of $[\bzero, \cdots \bzero, \bg_{N}]^\top$, backstepping creates additional directions through $\bg_i$ by interpreting $\bx_i$ as virtual control inputs. 

\subsection{Dual Relative Degree Systems}
In this section, we extend our backstepping approach to a different class of systems, considered in \cite{BahatiCDC25} and referred to as \emph{dual relative degree (DRD) systems}. Here, we consider a disturbed version of these systems, which have dynamics:
\begin{equation}\label{eq:DRD-system}
    \begin{aligned}
        \dot{\bz} = & \bf_{\bz}(\bz) + \bg_{\bz}(\bz)\bpsi(\bfeta)\bu_{\bz} + \bw_{\bz}(\bz)\bd,\\ 
        \dot{\bfeta} = & \bf_{\bfeta}(\bfeta) + \bg_{\bfeta}(\bfeta)\bu_{\bfeta} + \bw_{\bfeta}(\bfeta)\bd,
    \end{aligned}
\end{equation}
where $\bz\in\R^{n_1}$ and $\bu_{\bz}\in\R^{m_1}$ denote the state and input of the top layer while  $\bfeta\in\R^{n_2}$ and $\bu_{\bfeta}\in\R^{m_2}$ denote the state and input of the bottom layer. The dynamics of the top layer are \emph{almost} decoupled from the bottom, except for that this layer's control directions are influenced by $\bfeta$ via the function $\bpsi\,:\,\R^{n_2}\rightarrow\R^{r\times m_1}$. These systems are reminiscent of differentially flat systems, such as quadrotors \cite{MellingerICRA2011}. 

\begin{example}\label{ex:quadrotor}
The dynamics of a planar quadrotor are:
\begin{equation}\label{eq:quadrotor-dyn}
    \underbrace{
    \begin{bmatrix}
        \dot{x} \\ \dot{y} \\ \dot{\theta} \\ \ddot{x} \\ \ddot{y} \\ \ddot{\theta}
    \end{bmatrix}}_{\dot{\bx}}
    =
    \underbrace{
    \begin{bmatrix}
        \dot{x} \\ \dot{y} \\ \dot{\theta} \\ 0 \\ -g \\ 0
    \end{bmatrix}}_{\bf(\bx)}
    +
    \underbrace{
    \begin{bmatrix}
        0 & 0 \\ 0 & 0 \\ 0 & 0 \\ -\frac{1}{m}\sin(\theta) & 0\\ \frac{1}{m}\cos(\theta) & 0 \\ 0 & \frac{1}{J}
    \end{bmatrix}}_{\bg(\bx)}
    \underbrace{
    \begin{bmatrix}
        T \\ M
    \end{bmatrix}}_{\bu},
\end{equation}
where the state $\bx\in\R^6$ consists of the position $(x,y)$ and orientation $\theta$ of the quadrotor along with their rates, and the control input is the total thrust $T$ and moment $M$ generated by the propellers. Here, $m$, $g$, and $J$ denote the mass, acceleration to due gravity, and the moment of inertia. This system is in the form of \eqref{eq:DRD-system} with $\bz=(x,y,\dot{x},\dot{y})\in\R^4$, $\bfeta=(\theta,\dot{\theta})\in\R^2$, $\bu_{\bz}=T\in\R$, $\bu_{\bfeta}=M\in\R$.~\hfill$\bullet$ 
\end{example}
Similar to Sec. \ref{sec:strict-feedback}, we consider a safety constraint defined on the top layer states $\bz$ of \eqref{eq:DRD-system} as:
\begin{equation}\label{eq:Sz}
    \mathcal{S}_{\bz} = \{\bz\in\R^{n_1}\mid h_{\bz}(\bz) \geq 0\}.
\end{equation}
In contrast to strict-feedback systems, where $\bfeta$ would be viewed as a virtual control input to the top layer, we view $\bpsi(\bfeta)\bu_{\bz} = \bv\in\R^r$
as the virtual input to the top layer. Using this virtual input, we design a smooth feedback controller $\bk_{\bv}\,:\,\mathcal{D}_{\bz}\rightarrow\R^r$, with $\mathcal{D}_{\bz}\supset\mathcal{S}_{\bz}$, for the top layer satisfying: 
\begin{multline}\label{eq:DRD-virtual-controller}
    L_{\bf_{\bz}}h_{\bz}(\bz) + L_{\bg_{\bz}}h_{\bz}(\bz)\bk_{\bv}(\bz)  \\
    > -\theta_{d}\alpha(h_{\bz}(\bz)) + \tfrac{1}{\varepsilon}\|L_{\bw_{\bz}}h_{\bz}(\bz)\|^2.
\end{multline}
This controller can be constructed using the techniques mentioned in Sec. \ref{sec:strict-feedback}.
To align $\bu_{\bz}$with the virtual input, assume that $\bpsi$ has full column rank for all $\bfeta$ and define:
\begin{equation}\label{eq:kz}
    \begin{aligned}
        \bk_{\bz}(\bz,\bfeta) \coloneqq \argmin_{\bu\in\R^{m_1}}\|\bk_{\bv}(\bz) - \bpsi(\bfeta)\bu\|^2 = \bpsi(\bfeta)^{\dagger}\bk_{\bv}(\bz),
    \end{aligned}
\end{equation}
where $\bpsi(\bfeta)^{\dagger}$ denotes the left pseudoinverse. Taking $\bu_{\bz}=\bk_{\bz}(\bz,\bfeta)$ produces the partial closed-loop dynamics:
\begin{equation}\label{eq:DRD-system-partial}
    \begin{aligned}
        \dot{\bz} = & \bf_{\bz}(\bz) + \bg_{\bz}(\bz)\bpsi(\bfeta)\bk_{\bz}(\bz,\bfeta) + \bw_{\bz}(\bz)\bd, \\ 
        \dot{\bfeta} = & \bf_{\bfeta}(\bfeta) + \bg_{\bfeta}(\bfeta)\bu_{\bfeta} + \bw_{\bfeta}(\bfeta)\bd,
    \end{aligned}
\end{equation}
and ensures that:
\begin{equation}
    \bpsi(\bfeta)\bu_z = \bpsi(\bfeta)\bk_{\bz}(\bz,\bfeta) = \bpsi(\bfeta)\bpsi(\bfeta)^{\dagger}\bk_{\bv}(\bz).
\end{equation}
To proceed, we require the following assumption.
\begin{assumption}\label{assumption:etad}
    Given the smooth feedback controller $\bk_{\bv}\,:\,\R^{n_1}\rightarrow\R^r$ satisfying \eqref{eq:DRD-virtual-controller} there exists a smooth function $\bfeta_{d}\,:\,\R^{r}\rightarrow\R^{n_2}$ such  for all $\bz\in\mathcal{D}_{\bz}$:
    \begin{equation}
        \bpsi(\bfeta_{\rm{d}}(\bk_{\bv}(\bz)))\bpsi(\bfeta_d(\bk_{\bv}(\bz)))^{\dagger}\bk_{\bv}(\bz) = \bk_{\bv}(\bz).
    \end{equation}
\end{assumption}
The above assumption states that given any virtual input $\bv=\bk_{\bv}(\bz)$ generated by \eqref{eq:DRD-virtual-controller}, there exists a state of the bottom layer $\bfeta_{{d}}$ that aligns the control directions of the top layer $\bpsi(\bfeta_{d})$ with this virtual input in the sense that:
\begin{equation}
    \bpsi(\bfeta_{d})\bk_{\bz}(\bz,\bfeta_{d}) = \bk_{\bv}(\bz).
\end{equation}
This function must be constructed on a system-by-system basis; however, for certain systems, the existence of $\bfeta_{d}$ is tightly linked to differential flatness properties. For example, for the quadrotor \eqref{eq:quadrotor-dyn}, this function is:
\begin{equation}\label{eq:etad-drone}
    \bfeta_{d}(\bv) = 
    \begin{bmatrix}
        \mathrm{atan}\left(\frac{-v_1}{v_2} \right) & \odv{}{t}\mathrm{atan}\left(\frac{-v_1}{v_2} \right)
    \end{bmatrix}\T,
\end{equation}
which converts desired acceleration commands, represented by $\bv=\bk_{\bv}(\bz)$, into desired orientations and angular rates of the quadrotor, cf. \cite{MellingerICRA2011}.
Provided that $\bfeta=\bfeta_d$, the derivative of $h_{\bz}$ along the top layer dynamics satisfies:
\begin{equation*}
    \begin{aligned}
        \dot{h}_{\bz} = & L_{\bf_{\bz}}h_{\bz}(\bz) 
        + L_{\bg_{\bz}}h_{\bz}(\bz)\bpsi(\bfeta_{d})\bk_{\bz}(\bz,\bfeta_{d}) + L_{\bw_{\bz}}h_{\bz}(\bz)\bd \\ 
        = & L_{\bf_{\bz}}h_{\bz}(\bz) 
        + L_{\bg_{\bz}}h_{\bz}(\bz)\bk_{\bv}(\bz) + L_{\bw_{\bz}}h_{\bz}(\bz)\bd \\
        > & -\theta_{d}\alpha(h_{\bz}(\bz)) + \tfrac{1}{\varepsilon}\|L_{\bw_{\bz}}h_{\bz}(\bz)\|^2 - \|L_{\bw_{\bz}}h_{\bz}(\bz)\|\|\bd\|_{\infty} \\ 
        > & -\theta_{d}\alpha(h_{\bz}(\bz)) - \frac{\varepsilon}{2}\|\bd\|_{\infty}^2,
    \end{aligned}
\end{equation*}
thereby ensuring ISSf of $\mathcal{S}_{\bz}$ as in the proof of Theorem \ref{theorem:OD-ISSf-safety}. Hence, if $\bfeta\rightarrow\bfeta_{d}$, one may expect to be able to establish some form of safety guarantees for the overall system \eqref{eq:DRD-system}. To formalize this idea, consider the Lyapunov-like function:
\begin{equation}
    V(\bz,\bfeta) = \tfrac{1}{2\mu}\|\bfeta - \bfeta_{d}(\bk_{\bv}(\bz))\|^2,
\end{equation}
which we use to propose the CBF candidate:
\begin{equation}\label{eq:DRD-CBF}
    h(\bz,\bfeta) = h_{\bz}(\bz) - V(\bz,\bfeta),
\end{equation}
for \eqref{eq:DRD-system}. The following theorem provides conditions under which \eqref{eq:DRD-CBF} is an OD-ISSf-CBF for \eqref{eq:DRD-CBF}. 

\begin{theorem}\label{thm:DRD_safety}
    Consider system \eqref{eq:DRD-system} and a constraint set $\mathcal{S}_{\bz}$ as in \eqref{eq:Sz}. Suppose there exists a smooth controller $\bk_{\bv}\,:\,\mathcal{D}_{\bz}\rightarrow\R^r$ satisfying \eqref{eq:DRD-virtual-controller} and a smooth function $\bfeta_{d}\,:\,\R^r\rightarrow\R^{n_2}$ satisfying Assumption \ref{assumption:etad}, and construct $h$ as in \eqref{eq:DRD-CBF} with a set $\mathcal{D}\subseteq\mathcal{D}_{\bz}\times\R^{n_2}$ as in \eqref{eq:open-D}. If $\bg_{\bfeta}\,:\,\R^{n_2}\rightarrow\R^{n_2\times m_2}$ is full row rank and $\bpsi\,:\,\R^{n_2}\rightarrow\R^{r\times m_1}$ is full column rank for all $\bfeta\in\R^{n_2}$, then $h$ is an OD-ISSf-CBF for \eqref{eq:DRD-system}.
\end{theorem}

\begin{proof}
    We begin by computing the gradients of 
    % $V$ as:
    % \begin{equation*}
    %     \begin{aligned}
    %         \pdv{V}{\bfeta}(\bfeta,\bz) = & \frac{1}{\mu}(\bfeta - \bfeta_{d}(\bk_{\bv}(\bz)))\T \\ 
    %         \pdv{V}{\bz}(\bfeta,\bz) = & -\frac{1}{\mu}(\bfeta - \bfeta_{d}(\bk_{\bv}(\bz)))\T\pdv{\bfeta_{d}}{\bv}(\bk_{\bv}(\bz))\pdv{\bk_{\bz}}{\bz}(\bz), \\ 
    %     \end{aligned}
    % \end{equation*}
    % and those of 
    $h$ as:
    \begin{equation*}
        \begin{aligned}
            \pdv{h}{\boldeta}(\bz,\bfeta) = & -\frac{1}{\mu}(\bfeta - \bfeta_{d}(\bk_{\bv}(\bz)))\T \\ 
            \pdv{h}{\bz}(\bz,\bfeta) = & \pdv{h_{\bz}}{\bz}(\bz) \\ + & \frac{1}{\mu}(\bfeta - \bfeta_{d}(\bk_{\bv}(\bz)))\T\pdv{\bfeta_{d}}{\bv}(\bk_{\bv}(\bz))\pdv{\bk_{\bz}}{\bz}(\bz).
        \end{aligned}
    \end{equation*}
    Set $\bu_{\bz}=\bk_{\bz}(\bz,\bfeta)$ as in \eqref{eq:kz}, which exists when $\bpsi$ has full column rank, producing the partial closed-loop system as in \eqref{eq:DRD-system-partial}. Computing the Lie derivative of $h$ along the control directions of \eqref{eq:DRD-system-partial}, denoted by $\bg_{p}$, yields:
    \begin{equation*}
        \begin{aligned}
            L_{\bg_p}h(\bz,\bfeta) = & 
            \begin{bmatrix}
                \pdv{h}{\bz}(\bz,\bfeta) & \pdv{h}{\boldeta}(\bz,\bfeta)
            \end{bmatrix}
            \begin{bmatrix}
                \bzero \\ \bg_{\bfeta}(\bfeta)
            \end{bmatrix} \\ 
            = & -\frac{1}{\mu}(\bfeta - \bfeta_{d}(\bk_{\bv}(\bz)))\T\bg_{\bfeta}(\bfeta).
        \end{aligned}
    \end{equation*}
    If $\bg_{\bfeta}$ has full row rank, then when $L_{\bg_p}h(\bz,\bfeta)=\bzero$, we necessarily have
    % :
    % \begin{equation}\label{eq:eta=etad}
    %     \bfeta = \bfeta_{d}(\bk_{\bv}(\bz)).
    % \end{equation}
    $\bfeta = \bfeta_{d}(\bk_{\bv}(\bz))$.
    At such points, $\pdv{h}{\bfeta}(\bz,\bfeta)=\bzero$ and $\pdv{h}{\bz}(\bz,\bfeta)=\pdv{h_{\bz}}{\bz}(\bz)$, implying that the Lie derivatives along the drift dynamics of \eqref{eq:DRD-system-partial}, denoted by $\bf_{p}$, are:
    \begin{equation}\label{eq:Lfph-bound1}
        \begin{aligned}
            L_{\bf_p}h(\bz,\bfeta) = & L_{\bf_{\bz}}h_{\bz}(\bz) + L_{\bg_{\bz}}h_{\bz}(\bz)\bpsi(\bfeta)\bk_{\bz}(\bz, \bfeta) \\
            = & L_{\bf_{\bz}}h_{\bz}(\bz) + L_{\bg_{\bz}}h_{\bz}(\bz)\bpsi(\bfeta_d)\bk_{\bz}(\bz, \bfeta_d) \\ 
            = & L_{\bf_{\bz}}h_{\bz}(\bz) + L_{\bg_{\bz}}h_{\bz}(\bz)\bk_{\bv}(\bz) \\ 
            > & -\theta_{d}\alpha(h_{\bz}(\bz)) + \tfrac{1}{\varepsilon}\|L_{\bw_{\bz}}h_{\bz}(\bz)\|^2,
        \end{aligned}
    \end{equation}
    where the second line follows from $\bfeta = \bfeta_{d}(\bk_{\bv}(\bz))$, the third from Assumption \ref{assumption:etad}, and the fourth from \eqref{eq:DRD-virtual-controller}.
    After noting that when $L_{\bg_p}h(\bz,\bfeta)=\bzero$ we must have $h(\bz,\bfeta) = h_{\bz}(\bz)$ and $L_{\bw_p}h(\bz,\bfeta) = L_{\bw_{\bz}}h_{\bz}(\bz)$,
    we may express \eqref{eq:Lfph-bound1} as:
    \begin{equation}
        \begin{aligned}
            L_{\bf_p}h(\bz,\bfeta) > & -\theta_{d}\alpha(h(\bz,\bfeta)) + \tfrac{1}{\varepsilon}\|L_{\bw_p}h(\bz,\bfeta)\|^2.
        \end{aligned}
    \end{equation}
    From Lemma \ref{lemma:OD-ISSf-verify}, the above is precisely the statement that $h$ as in  \eqref{eq:DRD-CBF} is an OD-ISSf-CBF for the \emph{partial} closed-loop system \eqref{eq:DRD-system-partial}. Consequently, there exists a locally Lipschitz feedback controller $\bk_{\bfeta}(\bz,\bfeta)$ enforcing the OD-ISSf-CBF condition on \eqref{eq:DRD-system-partial}. Since this partial closed-loop system was obtained by taking the particular choice of input $\bu_{\bz}=\bk_{\bz}(\bz,\bfeta)$, this also implies the the existence of inputs for the original DRD-system \eqref{eq:DRD-system} satisfying the associated OD-ISSf-CBF condition for $h$. Hence, $h$  is an OD-ISSf-CBF for \eqref{eq:DRD-system}.
\end{proof}

\begin{remark}\label{remark:DRD}
    Strict-feedback and DRD systems can
    % in principle, 
    be combined and extended to yield a richer class of systems amenable to 
    % ISSf-CBF 
    backstepping. 
    For instance, although the quadrotor 
    % system 
    in Ex.~\ref{ex:quadrotor} is DRD, it does not satisfy the full row-rank condition required in Thm.~\ref{thm:DRD_safety}. However, when viewed as a multi-layer hybrid of strict-feedback and DRD, it does, and backstepping can be applied. The development of a unified framework accommodating multiple layers of both structures is omitted for brevity and will appear elsewhere.~\hfill$\bullet$
\end{remark}

\section{Examples}
% \subsection{Inverted Pendulum}
\noindent\textbf{Inverted Pendulum:}
To illustrate the developed techniques, we
% begin by 
consider an inverted pendulum with disturbances:
\begin{equation}\label{eq:pendulum}
    \dot{\bx} = 
    \begin{bmatrix}
        \dot{q} \\ \frac{g}{l}\sin(q) - \beta\dot{q}
    \end{bmatrix}
    +
    \begin{bmatrix}
        0 \\ \frac{1}{ml^2}
    \end{bmatrix}
    \bu 
    +
    \begin{bmatrix}
        \nu \\ \frac{1}{ml^2}
    \end{bmatrix}
    \bd,
\end{equation}
where $\bx=(q,\dot{q})\in\R^2$ consists of the pendulum's angle and angular velocity, $\bu\in\R$ is the control torque, $\bd\in\R$ is a disturbance input, and $g,l,m,\beta,\nu$ are physical parameters.
Our main objective is to design a feedback controller that ensures $q$ satisfies $|q| \leq 1$, 
despite disturbances. The disturbance $\bd$ is unmatched, but
the dynamics in \eqref{eq:pendulum} are in strict-feedback form,
allowing us to leverage backstepping to construct a CBF. At the top layer dynamics, the disturbance $\bd$ is matched since $\bw_1(q)=\nu\bg_1(q)=\nu$. The function $h_1(q) = 1 - q^2$
encodes the desired safety constraint and defines a set $\mathcal{S}_1$ as in \eqref{eq:S1}. This function satisfies $L_{\bg_1}h_1(\bx_1)=-2q$, which only vanishes inside $\mathcal{S}_1$, implying that $h_1$ is an OD-ISSf-CBF by Prop. \ref{prop:Lgh=0}. This allows for designing a smooth controller satisfying \eqref{eq:OD-ISSf-controller-smooth}, which we accomplish using the ``Half-Sontag" formula from \cite{CohenLCSS23}, and leads to the OD-ISSf-CBF from \eqref{eq:h-backstepping}. The resulting safe set in shown in Fig. \ref{fig:pendulum-expansion} (right) along with the expanded safe set for different values of $\delta$. Using this OD-ISSf-CBF to construct a QP based controller as in \eqref{eq:OD-ISSf-CBF-QP} for different values of $\varepsilon$ leads to the results in Fig. \ref{fig:pendulum-expansion} (left),  where this controller filters a nominal tracking controller and the system is subject to the disturbance signal $\bd(t) = \sin(t)$. For larger values of $\varepsilon$, this disturbance causes the trajectory to leave the safe set, but remain in bounded neighborhood of the set in accordance with Theorem \ref{theorem:OD-ISSf-safety}, whereas smaller values of $\varepsilon$ robustify the system to unmatched uncertainties and cause the system to evolve in $\mathcal{S}$.

\begin{figure}
    \centering \includegraphics{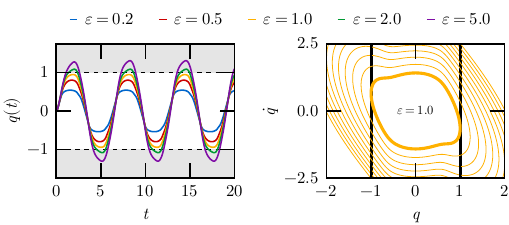}
    \vspace{-8mm}
    \caption{\textbf{Left}: Evolution of the pendulum's position under an ISSf safety filter, where the dashed black lines dente the cosntraint boundary. \textbf{Right}: Expansions of the pendulum's safe set using a virtual controller with $\varepsilon=1$. Here, the thick yellow curve corresponds to the boundary of $\mathcal{S}$ whereas the thinner curves corresponds to the boundary of $\mathcal{S}_{\delta}$ for different values of $\delta$.}
    \label{fig:pendulum-expansion}
    % \vspace{-1.0cm}
\end{figure}

% \subsection{Planar Quadrotor}
\noindent\textbf{Planar Quadrotor:}
To illustrate our results for DRD systems, consider a disturbed version of the planar quadrotor from  Ex. \ref{ex:quadrotor} with $\bd\in\R^2$ representing unknown wind gusts acting on the quadrotor's center of mass. This represents an unmatched uncertainty as such disturbances can only be directly canceled if the thrust vector  $\bpsi(\theta)=(-\sin(\theta),\cos(\theta))$ is aligned with $\bd$.  Keeping in line with
% the conditions of 
Thm. \ref{thm:DRD_safety}, we leverage a simplified version of the quadrotor, where the angular velocity $\dot{\theta}=\omega$ is viewed directly as a control input, yielding $\bfeta=\theta$, $\bu_{\bfeta}=\omega$, $\bf_{\bfeta}(\bfeta)=0$, and $\bg_{\bfeta}(\bfeta)=1$ in \eqref{eq:DRD-system}, which has full row rank\footnote{As noted in Remark \ref{remark:DRD} one could then backstep from $\omega$ to the input $M$ from \eqref{eq:quadrotor-dyn}, but such a presentation is omitted here in the interest of space.}. The dynamics of the $\bz$ subsystem are as in Ex. \ref{ex:quadrotor} with $\bf_{\bz}(\bz)=(\dot{x},\dot{z},0,-g)\in\R^4$, $\bg_{\bz}(\bz) = [\bzero; \tfrac{1}{m}\bI]\in\R^{4\times 2}$ and $\bw_{\bz}=\bg_{\bz}$. Our objective is to design a controller ensuring this system satisfies 
% the state constraint 
$x_{\max} \geq x$, representing a wall located at $x=x_{\max}$. We convert this constraint into a relative degree one function using high-order CBFs \cite{WeiTAC22,TanTAC22} to produce a function $h_{\bz}$ defining $\mathcal{S}_{\bz}$ as in \eqref{eq:Sz}. This is used to design a smooth virtual controller $\bk_{\bv}$ satisfying \eqref{eq:DRD-virtual-controller} using the Half-Sontag formula as in the previous subsection. We construct $\bk_{\bz}$ as in \eqref{eq:kz} and take $\bfeta_{d}$ as the first component of \eqref{eq:etad-drone}. Putting these together leads to the CBF \eqref{eq:DRD-CBF}, which satisfies the conditions of Thm. \ref{thm:DRD_safety}.

Simulations 
% of this system 
are performed using $\bu_{\bz}=\bk_{\bz}$ and $\bu_{\bfeta}=\bk_{\bfeta}$ as the controller satisfying the CBF condition for the partial closed-loop system \eqref{eq:DRD-system-partial}, whose safety properties are established in Thm. \ref{thm:DRD_safety}. Fig. \ref{fig:drone-sim1} illustrates the evolution of the system for different values of $\varepsilon$ and the disturbance $\bd=(1,0)$ pushing \emph{directly} toward unsafe states. Higher values of $\varepsilon$ lead to safety violations whereas lower values robustify the system to disturbances and ensure safety. As the drone approaches the boundary of the safe set, it orients itself (Fig. \ref{fig:drone-sim1}, right) so that its thrust is pointing away from the constraint boundary, allowing it to counteract disturbances that could force it to leave the safe set. Further simulations are shown in Fig. \ref{fig:drone-sim2} for different disturbance inputs using $\varepsilon=1$, which ensures safety of the system.

\begin{figure}
    \centering
    \includegraphics{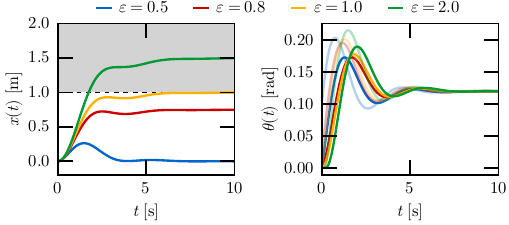}
    \vspace{-8mm}
    \caption{\textbf{Left}: Evolution of the drone's $x$ position under the DRD-CBF-based controller for different values of $\varepsilon$ and disturbance input $\bd=(1,0)$, where the gray region denotes unsafe states. \textbf{Right}: Evolution of the drone's orientation under the DRD-CBF-based controller for different values of $\varepsilon$, where the transparent lines denote the values of $\bfeta_{d}$.}
    \label{fig:drone-sim1}
\end{figure}

\begin{figure}
    \centering
    \includegraphics{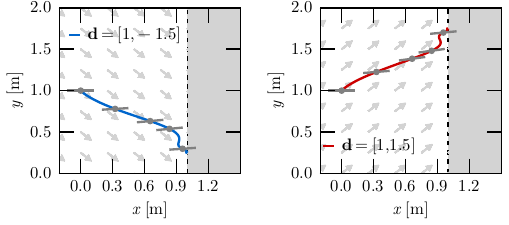}
    \vspace{-8mm}
    \caption{Evolution of the drone's $(x,y)$ position under different disturbances, where the gray arrows indicate the direction of the disturbance.}
    \label{fig:drone-sim2}
    \vspace{-0.7cm}
\end{figure}

\section{Conclusions}\label{sec:conclusions}
We presented a constructive framework for robust  safety-critical control of systems with unmatched disturbances. By extending ISSf via OD-CBFs and introducing a backstepping-based synthesis procedure, we relaxed verification requirements and enabled ISSf guarantees for broad system classes. Future work will include a full characterization of systems combining strict-feedback and DRD structure.

\bibliographystyle{ieeetr}
\bibliography{biblio}

\end{document}